\documentclass[natbib]{svjour3}                       
\smartqed  
\usepackage{lineno,hyperref}
\usepackage{amsmath}
\usepackage[tight,footnotesize]{subfigure}
\modulolinenumbers[5]
\usepackage{graphicx}
\usepackage{epstopdf}
\usepackage{relsize}
\usepackage{amssymb}
\newcommand*{\QEDA}{\hfill\ensuremath{\square}}%
\usepackage{mathptmx}      
\begin{document}

\title{Performance analysis of an unreliable $M/G/1$ retrial queue with two-way communication}

\titlerunning{Unreliable $M/G/1$ retrial queue with two-way communication}        

\author{Muthukrishnan Senthil Kumar	\and
        Aresh Dadlani	\and 
        Kiseon Kim 	
}



\date{Received: date / Accepted: date}

\maketitle

\begin{abstract}
Efficient use of call center operators through technological innovations more often come at the expense of added operation management issues. In this paper, the stationary characteristics of an $M/G/1$ retrial queue is investigated where the single server, subject to active failures, primarily attends incoming calls and directs outgoing calls only when idle. The incoming calls arriving at the server follow a Poisson arrival process, while outgoing calls are made in an exponentially distributed time. On finding the server unavailable (either busy or temporarily broken down), incoming calls intrinsically join the virtual orbit from which they re-attempt for service at exponentially distributed time intervals. The system stability condition along with probability generating functions for the joint queue length distribution of the number of calls in the orbit and the state of the server are derived and evaluated numerically in the context of mean system size, server availability, failure frequency and orbit waiting time.	
\keywords{Retrial queueing system \and Server breakdown \and Coupled switching \and Performance evaluation \and Steady-state distribution}
\subclass{60K25 \and 62N05}
\end{abstract}

\section{Introduction}
\label{section1} 
Blended call centers have recently evolved as an effective and profitable communication asset in bridging companies and their customers. Unlike conventional call centers, such modern communication systems are capable of managing a mixture of both, inbound and outbound call operations that require instant service \citep{Bhulai2003, Aksin2007}. An outgoing call is initiated by the server only when no incoming call is in the system. This feature, commonly referred to as \emph{coupled switching} or \emph{two-way communication}, yields higher productivity by reducing the idle time experienced by the serving operator \citep{Artalejo2012, Legros2017}. Moreover, incoming calls that find the server busy enter a virtual orbit and tend to retry for service after some random time \citep{Artalejo2008}. As a result, in-depth analysis of the influence of retrying customer calls on the dynamics of coupled switching in call centers is of great significance not only to the research community, but also serves as guidelines to statistical practitioners, network managers and system administrators. 

Current advancement in scale and scope of call centers as socio-technical systems has initiated the need for formulation and analysis of more refined queueing models. The range of seminal works dedicated to coupled switching in the retrial queueing literature is relatively diverse (see \cite{Artalejo2010b} for a comprehensive overview). In the study conducted by \cite{Choi1995}, some expected performance measures for an $M/G/1/K$ priority retrial queue with coupled switching were derived under the assumption that incoming and outgoing calls follow the same service time distribution. Nevertheless, such an assumption limits the practicality of the model as customers may have different service needs. Although \cite{Artalejo2010} derived the first partial moments for an $M/G/1/1$ retrial queue model with general service time distributions using mean value analysis, it cannot be used to obtain the stationary distribution and factorial moments. Comprehensive analysis of the $M/M/1/1$ retrial queue with coupled switching and different service time distributions for single and multiple server cases have been reported by \cite{Artalejo2012}. The work was further extended to incorporate multiple types of outgoing calls by \cite{Sakurai2015}, for which the joint stationary distribution of the number of calls in the orbit and the state of the server were obtained both, asymptotically and recursively. Furthermore, \cite{Artalejo2013} proposed an embedded Markov chain approach to study the steady-state behavior of a couple-switched $M/G/1$ retrial queue with tailed asymptotic analysis of number of customers residing in the orbit. Nonetheless, the server may undergo multiple failures and resume service upon repair (\cite{Krishnamoorthy2014}). Despite its prevalence in practice, research efforts to address server failure in blended call center system models are scarce.

The reliability of an $M/G/1$ retrial queue with only inbound calls and server breakdowns was investigated by \cite{Wang2001}. In another related work by \cite{Martin1995}, an $M/G/1$ service system model with two types of impatient units was exhaustively analyzed. The equilibrium balking strategies of customers in the $M/M/1$ queue with set-up times, breakdowns and reparis were scrutinized by \cite{Chen2015}. Moreover, \cite{Ouazine2016} reported a functional approximation of the stationary performance of the $M1,M2/G1,G2/1$ retrial queue with two-way communication and finite orbit capacity. Perfect and imperfect repair of a single server $M/M/1/1$ retrial queue with only incoming customers were modeled by \cite{Chang2017}. A more recent work by \cite{Phung2017} successfully identified higher order moments for single server retrial queues with set-up time using Taylor series method. To our best knowledge, however, explicit reliability indices for an unreliable $M/G/1$ retrial queue with two-way communication have not yet been analytically derived in modeling service systems. Hence, continuous-time analytical characterization of an unreliable single server retrial queue with coupled switching is imperative from the viewpoint of both, queueing as well as reliability analysis.

The foremost goal of this letter is to study the impact of server failure on the steady-state performance of the $M/G/1$ queue with two-way communication having an orbit with infinite capacity and generally distributed server repair time. In particular, we obtain the system stability condition using the embedded Markov chain technique, followed by the supplementary variable approach to obtain in closed-form the probability generating functions (pgfs) for incoming calls in orbit and in the system, followed by their second order moments. The numerical simulations conducted for various performance metrics of interest corroborate the theoretical findings of the proposed system model.

The rest of the paper is organized as follows: Sect.~\ref{section2} is dedicated to the description and mathematical formulation of the unreliable $M/G/1$ retrial queue with two-way communication. In Sect.~\ref{section3}, the necessary and sufficient condition for system stability is presented, followed by derivation of the corresponding steady-state distribution. Various system performance measures as well as reliability indices are derived in Sect.~\ref{section4}. Numerical simulation results are discussed in Sect.~\ref{section5}. Finally, Sect.~\ref{section6} concludes the paper with suggestions for potential future research directions.

\section{System model formulation}
\label{section2} 
We consider a single server retrial queue in which primary inbound calls follow a Poisson arrival process with rate $\lambda$. If the server is idle, an outgoing call is initiated after an exponentially distributed time with rate $\alpha$. As in reality, the time taken to serve incoming and outgoing calls is assumed to be different. If an incoming call finds the server busy, it then enters the orbit and re-attempts to seek service after an exponentially distributed time with rate $\nu$. Otherwise, the incoming call commences service immediately. Since the server may breakdown while serving calls, without loss of generality, we assume that the lifetime of the server follows an exponential distribution with rates $\beta_1$ and $\beta_2$ during the service of inbound and outbound calls, respectively. On failure, the server is instantly sent for repair which has a generally distributed time. The corresponding state transition diagram of the proposed model is given in \figurename{~\ref{fig1}}, where $n \geq 0$ is the number of incoming calls waiting in the orbit.
\begin{figure}[!t]
	\centering
	\includegraphics[width=3.6 in]{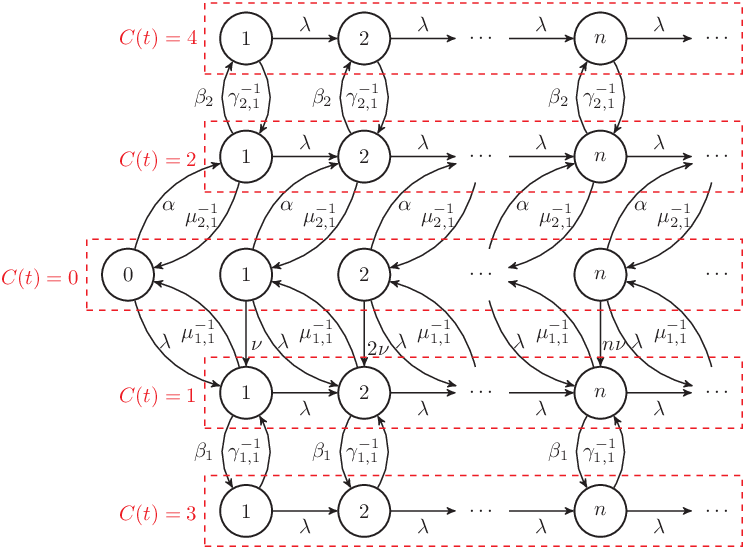}
	\caption{State transitions of the proposed system model with each row highlighting the state of the server.}
	\label{fig1}
\end{figure}

For the sake of consistency, we define $i \in \{1,2\}$ to differentiate between incoming and outgoing calls. Henceforth, $i=1$ refers to incoming calls, while $i=2$ indicates outgoing calls. Let $S_i(x)$ and $R_i(x)$ be the cumulative distributions of service and repair times of $i$-type calls, respectively. Similarly, let $s_i(x)$ and $r_i(x)$ denote respectively, the probability density functions of service and repair times of $i$-type calls. The Laplace transforms of the service and repair times for each type of call are denoted as $\tilde{S}_i(\theta)$ and $\tilde{R}_i(\theta)$, respectively. We also define $S^o_i(x)$ and $R^o_i(x)$ as the remaining service and repair times, respectively. Moreover, let $\mu_{i,k}$ and $\gamma_{i,k}$ denote the $k^{th}$ moment of service and repair times, respectively. In what follows, the arrival flows of incoming calls, outgoing calls, service time, repair time, and intervals between successive re-attempts are all assumed to be mutually independent. Finally, let $N(t)$ be the number of incoming customer calls in orbit and $M(t)$ be the total number of customers in the system at time $t$. We now define the state of the server, denoted by $C(t)$, to be as follows:
\begin{equation}
	 C(t) =
  \begin{cases}
    0,		& \quad \text{if the server is } idle, \\
    1,		& \quad \text{if the server is } busy~\text{serving an}~incoming~call, \\
    2,		& \quad \text{if the server is } busy~\text{making an}~outgoing~call,\\
    3,		& \quad \text{if the server } fails~\text{while serving an}~incoming~call,\\
    4,		& \quad \text{if the server } fails~\text{while making an}~outgoing~call.
  \end{cases} \nonumber
\end{equation}

For service and repair times that are exponentially distributed, the state transitions for $\{(C(t), N(t)); t \geq 0\}$ on the state space $S\!=\!\{0,1,2,3,4\} \!\times\! \mathbb{Z}_+$ are as shown in \figurename{~\ref{fig1}}, where $\mathbb{Z}_+$ denotes the set of non-negative integers. In the case of generally distributed service and repair times, the Markov process $\{(C(t), N(t), S^o_1(t), S^o_2(t), R^o_1(t), R^o_2(t)); t \geq 0\}$ can be used to describe the state of the system. Based on this generalized definition, the state probabilities are given by:
\begin{align}
&P_{0,n}(t) = \text{Pr}[C(t) = 0, ~N(t) = n]\,, \label{eq1}\\
&P_{1,n}(x,t)\,dx = \text{Pr}[C(t) = 1, ~N(t) = n, ~x < S^o_1(t) \leq x + dx]\,, \label{eq2}\\
&P_{2,n}(x,t)\,dx = \text{Pr}[C(t) = 2, ~N(t) = n, ~x < S^o_2(t) \leq x + dx]\,, \label{eq3}\\
&P_{3,n}(x,y,t)\,dy = \text{Pr}[C(t) = 3, ~N(t) = n, ~S^o_1(t) = x, ~y < R^o_1(t) \leq y + dy]\,, \label{eq4}\\
&P_{4,n}(x,y,t)dy = \text{Pr}[C(t) = 4, ~N(t) = n, ~S^o_2(t) = x, ~y < R^o_2(t) \leq y + dy]\,, \label{eq5}
\end{align}
where $x, y \geq 0$ are time epochs. In (\ref{eq1}), $P_{0,n}(t)$ is the probability of the server being idle while having $n$ calls in the orbit at time $t$. For $i \in \{1,2\}$ in (\ref{eq2}) and (\ref{eq3}), $P_{i,n}(x,t)\,dx$ denotes the joint probability that the server is busy with an $i$-type call during the remaining service time $(x, x + dx)$ and there are $n$ calls residing in the orbit at time $t$. Likewise, for $j \in \{3,4\}$ in (\ref{eq4}) and (\ref{eq5}), $P_{j,n}(x,y,t)\,dy$ refers to the joint probability that at time $t$ there are $n$ calls residing in the orbit, the remaining service time is $x$, and the failed server is fixed within the remaining repair time $(y, y+dy)$ while serving an inbound $(j = 3)$ or an outbound $(j = 4)$ call.

\section{Steady-state distribution}
\label{section3}
In this section, we identify the pgfs of orbit size and number of incoming calls in the system. To do so, we first determine the system stability condition using the following theorem.

\begin{theorem}
The necessary and sufficient condition for system stability is given by the inequality $\lambda \mu_{1,1}(1 + \beta_1 \gamma_{1,1}) < 1$.
\end{theorem}
\begin{proof}
Let $\hat{X}_n$ be the service completion time of the $n^{th}$ call which includes possible down times (due to server failure) while providing service. For the sufficient condition, we need to prove the ergodicity of $\{L_n ; n \geq 1\}$, where $\{L_n\}$ is an irreducible and aperiodic discrete-time Markov chain of $\{(C(t), N(t), S^o_1(t), S^o_2(t), R^o_1(t), R^o_2(t)); t \geq 0\}$ and is defined as $L_n\!=\! N(\hat{X}_n^+)$. Using Foster's criterion and undertaking the same approach as in \cite{Artalejo2013}, $\{L_n\}$ is positive recurrent if $|\eta_k| < \infty$ and $\lim_{k \to \infty} \sup \{\eta_k\} < 0$ for all $k$, where $\eta_k = E[(L_{n + 1} - L_n) | L_n \!=\! k]$. By conditioning on the identity of the $n^{th}$ call, we arrive at:
\begin{equation}
	\eta_k=\frac{kv [\lambda \mu_{1,1}(1\!+\!\beta_1\gamma_{1,1})\!-\!1]}{\lambda\!+\!kv\!+\!\alpha}\!+\!\frac{\lambda [\lambda \mu_{1,1}(1\!+\!\beta_1\gamma_{1,1})]}{\lambda\!+\!kv\!+\!\alpha}\!+\! \frac{\alpha [\lambda \mu_{2,1}(1\!+\!\beta_2 \gamma_{2,1})]}{\lambda\!+\!kv\!+\!\alpha}\,.
	\label{eq6}
\end{equation}

It is straightforward to observe that for all $k$ values, $\eta_k < \infty$ and $\lim_{k \to \infty} \sup \{\eta_k\} < 0$ if $\lambda \mu_{1,1} (1 + \beta_1 \gamma_{1,1}) < 1$, which proves the sufficiency criteria. 

As pointed out in \cite{Sennott1983}, the non-ergodicity of $\{L_n\}$ can be guaranteed if Kaplan's condition is satisfied, i.e. there exists some $k_0 \!\in\! Z_+$ such that $\eta_k \!\geq\! 0$ for $k \!\geq\!k_0$ and $\eta_k \!<\! \infty$ for all $k \!\geq\! 0$. In our setting, this condition is satisfied as $r_{i,j}=0$ for $j<i-1$, where $P=[r_{i,j}]$ is the one-step transition probability matrix. Hence, $\lambda \mu_{1,1} (1+\beta_1 \gamma_{1,1}) \geq 1 $ implies the non-ergodicity of  $\{L_n; n\geq 1\}$, which completes the proof. \QEDA
\end{proof}

Adopting the supplementary variable technique, the system of balance equations for (\ref{eq1})-(\ref{eq5}) is obtained in terms of the limiting state probabilities as follows:
\begin{align}
&(\lambda + n \nu + \alpha) P_{0,n} = P_{1,n}(0) + P_{2,n}(0)\,, \label{eq7}\\
&\!P'_{1,n}(x) \!=\!(\lambda \!+\! \beta_1)P_{1,n}(x) \!-\! \lambda P_{0,n}s_1(x) \!-\! \lambda P_{1,n - 1}(x) \!-\! (j \!+\! 1)\nu P_{0,n + 1} s_1(x) \!-\! P_{3,n}(x,0)\,, \label{eq8}\\
&\!P'_{2,n}(x) \!=\! (\lambda \!+\! \beta_2)P_{2,n}(x) \!-\! \lambda P_{2,n\!-\!1}(x) \!-\! \alpha P_{0,n} s_2(x) \!-\! P_{4,n}(x,0)\,, \label{eq9}\\
&\frac{\partial}{\partial y} P_{3,n}(x,y) \!=\! -\lambda [P_{3,n\!-\!1}(x,y) \!-\! P_{3,n}(x,y)] \!-\! \beta_1 P_{1,n}(x) r_1(y)\,, \label{eq10}\\
&\frac{\partial}{\partial y} P_{4,n}(x,y) \!=\! -\lambda [P_{4,n\!-\!1}(x,y) \!-\! P_{4,n}(x,y)] \!-\! \beta_2 P_{2,n}(x) r_2(y)\,, \label{eq11}
\end{align}
with the following normalizing condition, where the terms $W(x)$ and $Z(x,y)$ respectively, stand for $P_{1,n}(x) + P_{2,n}(x)$ and $P_{3,n}(x,y) + P_{4,n}(x,y)$:
\begin{equation}
	\mathlarger{‎‎\sum}\limits_{n=1}^{\infty} \Bigg[P_{0,n}+\int_{0}^{\infty} \! W(x)dx + \int_{0}^{\infty}\!\!\int_{0}^{\infty} Z(x,y)dx\, dy\Bigg] = 1 ,
	\label{eq12}
\end{equation}

Taking the Laplace transforms $\mathcal{L\{\cdot\}}$ of (\ref{eq7})-(\ref{eq11}) results in the following marginal generating functions, where notations $\tilde{P}_{i,n}(\theta)$ and $\widetilde{\tilde{P}}_{j,n}(\theta,s)$ are used to denote $\mathcal{L}\{P_{i,n}(x)\}$ and $\mathcal{L}\{\mathcal{L}\{P_{j,n}(x,y)\}\}$, respectively:
\begin{align}
&P_0(z) = \sum\limits_{n=0}^{\infty} P_{0,n}\,z^n\,, \label{eq13}\\
&\tilde{P}_i(z,\theta) = \sum\limits_{n=0}^{\infty} \tilde{P}_{i,n}(\theta)\,z^n\,,\qquad\qquad i \in \{1,2\}\,, \label{eq14}\\
&P_i(z,0) = \sum\limits_{n=0}^{\infty} P_{i,n}(0)\,z^n\,,\qquad\qquad~i \in \{1,2\}\,, \label{eq15}\\
&\widetilde{\tilde{P}}_j(z,\theta,s)= \sum\limits_{n=0}^{\infty} \widetilde{\tilde{P}}_{j,n}(\theta,s)\,z^n\,,\quad~~~j \in \{3,4\}\,, \label{eq16}\\
&\tilde{P}_{j}(z,\theta,0) = \sum\limits_{n=0}^{\infty}\tilde{P}_{j,n}(\theta,0)\,z^n\,,\quad~~~ j \in \{3,4\}\,. \label{eq17}
\end{align}

Subsequently, the following pgfs are obtained through some algebraic manipulations, with function definitions $\phi(z)\triangleq\exp\big(-\int_z^1 \frac{[\lambda (1-\delta_1(u)) + \alpha (1-\delta_2(u))]}{\nu [\delta_1(u)-u]} du\big)$, $\delta_i(\cdot) \triangleq \tilde{S}_i(h_i(\cdot))$ and $h_i(z) \triangleq \lambda + \beta_i - \lambda z - \beta_i \tilde{R}_i(\lambda - \lambda z)$ for $i \!\in\! \{1,2\}$:
\begin{align}
&P_{0}(z) = \frac{1 \!-\! \lambda \mu_{1,1}(1 \!+\! \beta_1 \gamma_{1,1})}{1 \!+\! \alpha \mu_{2,1}(1 \!+\! \beta_2 \gamma_{2,1})} \phi(z)\,, \label{eq18}\\
&\tilde{P}_{1}(z,0) = \frac{[\lambda(1 \!-\! z) \!+\! \alpha(1 \!-\! \delta_2(z))][1 \!-\! \delta_1(z)]}{\big[\delta_1(z)-z\big] h_1(z)} P_0(z)\,, \label{eq19}\\
&\tilde{P}_2(z,0) = \frac{\alpha [1 \!-\! \delta_2(z)]}{h_2(z)} P_0(z)\,, \label{eq20}\\
&\widetilde{\tilde{P}}_3(z,0,0) = \frac{\beta_1 [1 \!-\! \tilde{R}_1(\lambda \!-\! \lambda z)]}{\lambda \!-\! \lambda z} \tilde{P}_1(z,0)\,, \label{eq21}\\
&\widetilde{\tilde{P}}_4(z,0,0) = \frac{\beta_2 [1 \!-\! \tilde{R}_2(\lambda \!-\! \lambda z)]}{\lambda \!-\! \lambda z} \tilde{P}_2(z,0)\,. \label{eq22}
\end{align}

In steady-state, the pgfs of orbit occupancy size, $P(z)$, and system size, $R(z)$, at an arbitrary epoch can now be expressed in terms of the equations derived in (\ref{eq18})-(\ref{eq22}). By ignoring the non-zero probability of server failure, it is straightforward to show that the following results are in complete agreement with \cite{Artalejo2013}: 
\begin{align}
P(z) =& \,P_0(z) + \tilde{P}_1(z,0) + \tilde{P}_2(z,0) + \widetilde{\tilde{P}}_3(z,0,0) + \widetilde{\tilde{P}}_4(z,0,0) \nonumber\\
=& \,\frac{\lambda (1-z) + \alpha[1-\delta_2(z)]}{\lambda[\delta_1(z)-z]} P_0(z)\,, \label{eq23}\\
R(z) =& \,P_0(z) + z[\tilde{P}_1(z,0) + \tilde{P}_2(z,0) + \widetilde{\tilde{P}}_3(z,0,0) + \widetilde{\tilde{P}}_4(z,0,0)] \nonumber\\
=& \,\frac{(\lambda - \lambda z) \delta_1(z) + \alpha[1 - \delta_2(z)]}{\lambda[\delta_1(z) - z]} P_0(z)\,. \label{eq24}
\end{align}

\section{Performance and reliability analysis}
\label{section4}
In this section, some performance and reliability metrics for the queueing system under study are discussed. Specifically, our performance analysis involves finding the expected number of calls in the system and expected waiting time in the orbit, while server availability and server failure frequency characterize the reliability indices. To improve readability, we introduce and use the notations $\rho_1 \!=\! (1 + \beta_1 \gamma_{1,1})$, $\rho_2 \!=\! (1 + \beta
_2 \gamma_{2,1})$, $\sigma_1 \!=\! \lambda \rho_1 \mu_{1,1}$, and $\sigma_2 \!=\! \alpha \rho_2 \mu_{2,1}$ throughout this section.

\subsection{Expected number of customer calls in the system}
This measure accounts for the mean number of incoming calls retrying for service, either due to server failure or it being busy, as well as those being served by the server. This is readily obtained by differentiating the pgfs in (\ref{eq23}) and (\ref{eq24}) before evaluating them at $z\!=\! 1$. As a result, the equation given in (\ref{eq23}) yields the first moment of the orbit size as follows:
\begin{align}
E[N] =& P'(1) \nonumber\\
 =& \frac{\lambda^2(\beta_1 \mu_{1,1} \gamma_{1,2} + \rho_1^2 \mu_{1,2})}{2(1 - \sigma_1)} \!+\! 
\frac{\lambda \alpha (\beta_2 \mu_{2,1} \gamma_{2,2} + \rho_2^2 \mu_{2,2})}{2(1 + \sigma_2)} \!+\! 
\frac{\lambda (\sigma_1 + \sigma_2)}{\nu (1 - \sigma_1)}\,, \label{eq25}
\end{align}

Similarly, the mean system size resulting from (\ref{eq24}) is given as:
\begin{align}
E[M] = R'(1) = P'(1) + \frac{\sigma_1 + \sigma_2}{1 + \sigma_2}\,. \label{eq26}
\end{align}

By differentiating (\ref{eq23}) and (\ref{eq24}) twice with respect to $z$ and evaluating them at $z\!=\! 1$, we also obtain the second order moment of the orbit size as follows:
\begin{align}
E[N^2] = P''(1) = M_1+M_2+M_3\,,  \label{eq27}
\end{align}
where $M_1$, $M_2$, and $M_3$ are derived to be:
\begin{align}
M_1 = &\frac{\lambda^2(\sigma_1 + \sigma_2)^2 }{\nu^2(1 - \sigma_1)(1 + \sigma_2)} +
\frac{\lambda^3 (1 + \sigma_2)(\beta_1 \mu_{1,1} \gamma_{1,2} + \rho_1 \mu_{1,2})}{2 \nu (1 - \sigma_1) (1 + \sigma_2)} \qquad\qquad\qquad \nonumber\\
&\qquad\qquad\qquad\qquad+ \frac{\lambda^2 \alpha (1 - \sigma_1)(\beta_2 \mu_{2,1} \gamma_{2,2} + \rho_2 \mu_{2,2})}{2 \nu (1 - \sigma_1)(1 + \sigma_2)}\,, \label{eq28} \\
M_2 = &\frac{\lambda^3(\beta_1 \gamma_{1,3} \mu_{1,1} + 2 \rho_1 \beta_1 \gamma_{1,2} \mu_{1,2} + \rho_1^3 \mu_{1,3})}{3(1 - \sigma_1)} +
\frac{\lambda^4 (\beta_1 \mu_{1,1}\gamma_{1,2} + \rho_1 \mu_{1,2})^2}{2(1 - \sigma_1)^2} \nonumber\\
&\!+\! \frac{\lambda^3 (\beta_1 \mu_{1,1} \gamma_{1,2} + \rho_1 \mu_{1,2})}{2(1 - \sigma_1)^2} \Bigg(\frac{\sigma_1 + \sigma_2}{\nu} + \frac{\alpha\big(1 - \sigma_1(\beta_2 \mu_{2,1}\gamma_{2,2} + \rho_2 \mu_{2,2})\big)}{2(1 + \sigma_2)}\Bigg) \nonumber\\
&\!+\! \frac{\lambda^2 \sigma_1}{1 - \sigma_1} \Bigg(\frac{\alpha \rho_2\big(1 - \sigma_1(\beta_2 \gamma_{2,3} \mu_{2,1} + 2 \rho_2 \beta_2 \gamma_{2,2} \mu_{2,2} + \rho_2^3 \mu_{2,3})\big)}{3 (1 + \sigma_2)}\Bigg) \nonumber\\
&\!+\! \frac{\lambda^2 \sigma_1 (\sigma_1 + \sigma_2)(\beta_2 \mu_{2,1} \gamma_{2,2} + \rho_2 \mu_{2,2})}{\nu (1- \sigma_1)} + \frac{M_4}{\nu (1 - \sigma_1)}\,, \label{eq29}\\
M_3 = &\frac{\lambda^2 \alpha( 1 - \sigma_1)(\beta_1 \gamma_{1,3} \mu_{1,1} + 2 \rho_1 \beta_1 \gamma_{1,2} \mu_{1,2} + \rho_1^3 \mu_{1,3})}{3(1+\sigma_2)} \nonumber\\
&\!+\! \frac{\lambda^2 \alpha (\sigma_1 + \sigma_2)(\beta_2 \mu_{2,1} \gamma_{2,2} + \rho_1 \mu_{2,2})}{\nu (1 + \sigma_2)} + \frac{\lambda^2 \alpha \sigma_2 M_4}{\nu (1 - \sigma_1)(1 + \sigma_2)}\,, \label{eq30}
\end{align}
with $M_4$ given as below:
\begin{align}
M_4 = &\frac{(\sigma_1 + \sigma_2)^2}{v} + \frac{\lambda(1 + \sigma_2)(\beta_1 \mu_{1,1} \gamma_{1,2} + \rho_1 \mu_{1,2})}{2} + \frac{\alpha (1 - \sigma_1)(\beta_2 \mu_{2,1} \gamma_{2,2} + \rho_2 \mu_{2,2})}{2}\,. \label{eq31}
\end{align}

\subsection{Expected waiting time in orbit}
Denoted by $W$, the steady-state delay experienced by an incoming customer call in obit depends on the total idle time of the server not serving an incoming call ($W_0$), the total service time (including the server failure time) of the server providing service to an incoming call ($W_1$), and the total service time (including the server failure time) of the server busy with an outgoing call ($W_2$). The probability of an inbound call entering the orbit ($P_w$) is thus, calculated as follows:
\begin{align}
P_w &= \lim_{z \!\to\! 1} \{\tilde{P}_1(z,0) \!+\! \tilde{P}_2(z,0)\ \!+\! \widetilde{\tilde{P}}_3(z,0,0) \!+\! \widetilde{\tilde{P}}_4(z,0,0)\} = \frac{(\sigma_1 + \sigma_2)}{(1 + \sigma_2)}\,. \label{eq32}
\end{align}

Using the above equation and the first moments of the pgfs in (\ref{eq18})-(\ref{eq22}), we derive the mean waiting time in the orbit to be (see \cite{Choi1995}):
\begin{align}
	\text{E}[W] = \text{E}[W_0] + \text{E}[W_1] + \text{E}[W_2] = \text{E}[N]/\lambda\,, \label{eq33}
\end{align}
where $\text{E}[W_0]= P_w/\nu$, $\text{E}[W_1]=\text{E}[N] \text{E}[B_1] + \sigma_1 \text{E}[R_1]$, and $\text{E}[W_2] = \sigma_2 (1 - \sigma_1) \text{E}[R_2] / (1 + \sigma_2) + \sigma_2 \text{E}[W_0]$. Note that the notations $\text{E}[B_i]$ and $\text{E}[R_i]$ represent the mean service time (including failure time) and the mean remaining service time (including failure time) while serving $i$-type calls, which are given as $\mu_{i,1}(1+  \beta_i \gamma_{i,1})$ and $\text{E}[B_i^2]/(2 \text{E}[B_i])$, respectively.

\subsection{Server availability}
Availability is a system characteristic that measures how often the server is available for use, even though it may not be functioning properly. The probability that the server is operational at a given time instant $t$ is defined as its point-wise availability, $A(t)$, and its steady state availability (i.e. $\lim_{t \to \infty} A(t)\!=\!P_{a}$) is  given as:
\begin{align}
P_{a} &= \lim_{z \!\to\! 1} \{P_0(z) \!+\! \tilde{P}_1(z,0) \!+\! \tilde{P}_2(z,0)\} 
= \frac{(1 + \alpha \mu_{2,1})(1 - \sigma_1) + \lambda \mu_{1,1} (1 + \sigma_2)}{(1 + \sigma_2)}\,. \label{eq34}
\end{align}

\subsection{Server failure frequency}
This measure corresponds to the probability that the server fails at time $t>0$ given that it was operating at $t=0$ (see \cite{Kumar2010}). It can be easily shown that the following closed-form expression results from (\ref{eq23}):
\begin{align}
P_f &= \lim_{z \!\to\! 1} \{\beta_1 \tilde{P}_1(z,0) \!+\! \beta_2 \tilde{P}_2(z,0)\} 
= \lambda \mu_{1,1} \beta_1 + \alpha \mu_{2,1} \beta_2  \frac{(1 - \sigma_1)}{(1 + \sigma_2)}\,. \label{eq35}
\end{align}

\section{Numerical Examples and Discussions}
\label{section5}
To illustrate the impact of system parameters on the performance primitives, we present numerical examples for service and repair times with three arbitrary distributions namely, exponential with density function $c_1 e^{-c_1x}$, Erlangian of order two with density function $c_1^2 x e^{-c_1x}$ and hyperexponential given as $a c_1 e^{-c_1 x}\!+\!(1\!-\!a)c_2 e^{-c_2 x}$, where $c_1, c_2\!>\!0$ and $0\!\leq\!a\!\leq\!1$. Throughout this section, we assume $\lambda\!=\!1.2$, $\alpha\!=\!0.4$, $\nu\!=\!1$, $\mu_{2,1}\!=\!0.1$, and $\gamma_{2,1}\!=\!0.2$ to satisfy the ergodic condition of the analytical system. We also consider an $M/G/1$ retrial queue without server failure, i.e. $(\beta_1,\beta_2) = (0,0)$, as in \cite{Artalejo2013} to serve as the baseline scenario for our comparison.


\figurename{~\ref{fig2}} shows the variation in mean system size ($E[M]$) as functions of the inbound arrival rate ($\lambda$), outbound rate ($\alpha$), retrial rate ($\nu$), and inbound service ($\mu_{1,1}$) and repair ($\gamma_{1,1}$) times. As evident in \figurename{\ref{fig2a}}, increase in the number of arriving calls reduces the chance of finding the server active and idle. Consequently, these unattended incoming calls enter the orbit to retry for service thus, increasing the average system size as shown in the figure. In comparison to the failure-free baseline scenario, we note that the system size of our model increases with the failure rate $\beta_1$ as $\lambda$ increases. A similar relationship can be observed in \figurename{\ref{fig2b}} between $E[M]$ and $\alpha$ as well. On the other hand, $E[M]$ steeply decreases initially and gradually stabilizes to some constant value with increase in the rate for service retrial as depicted in \figurename{\ref{fig2c}}. The observation justifies the fact that an incoming call residing in orbit has a higher chance of being served and thus, leaving the orbit if it re-attempts for service more frequently. For lower values of $\beta_1$, a primary incoming call is more probable to find the server available, resulting in a reduced orbit size. The influence of the service and repair times of primary calls on the mean system size are also illustrated in \figurename{\ref{fig2d}} and \figurename{\ref{fig2e}}, respectively. As the average time to serve incoming calls increases, $E[M]$ grows steeper as depicted in \figurename{~\ref{fig2d}}. In other words, longer service times increases the number of incoming calls waiting in the orbit. Likewise, the shorter the server repair time, the more active it would be thus, reducing the system size which is mainly dominated by the orbit length. As seen in \figurename{\ref{fig2e}}, depending on the distribution type, the difference in $E[M]$ increases substantially with increase in the server repair time while serving an inbound call.
\begin{figure}[!t]
  \begin{center}
    \subfigure[Mean system size versus $\lambda$.]{\label{fig2a}\includegraphics[scale=0.55]{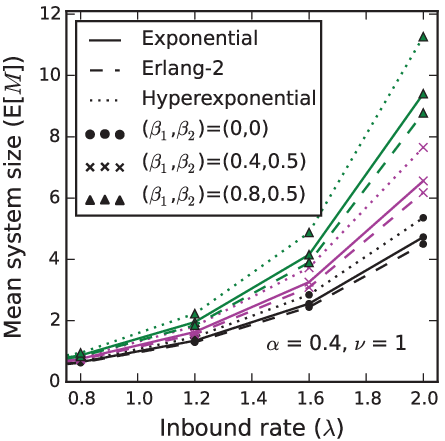}}~
    \subfigure[Mean system size versus $\alpha$.]{\label{fig2b}\includegraphics[scale=0.55]{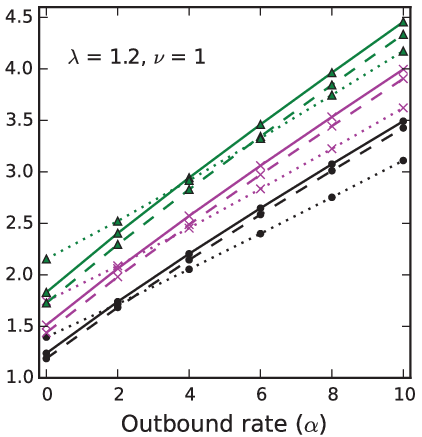}}~
   \subfigure[Mean system size versus $\nu$.]{\label{fig2c}
\includegraphics[scale=0.55]{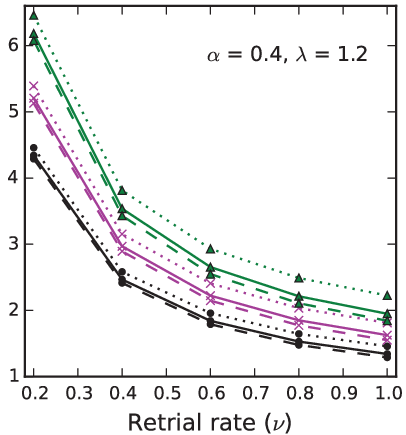}}\\
   \subfigure[Mean system size versus $\mu_{1,1}$.]{\label{fig2d}\includegraphics[scale=0.55]{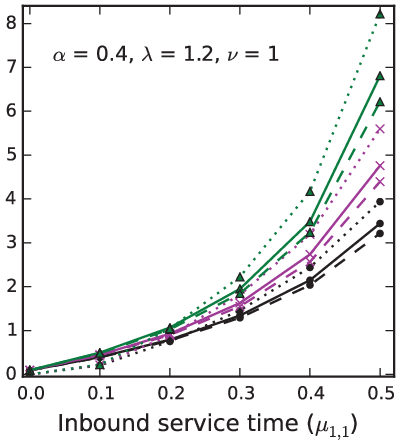}}~~~
   \subfigure[Mean system size versus $\gamma_{1,1}$.]{\label{fig2e}
\includegraphics[scale=0.55]{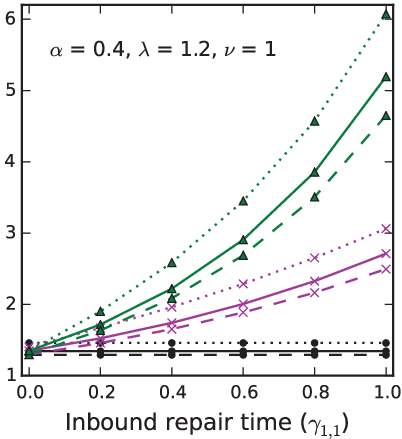}}
  \end{center}
  \caption{Mean system size ($E[M]$) against incoming arrival rate ($\lambda$), outgoing rate ($\alpha$), retrial rate ($\nu$), inbound service time ($\mu_{1,1}$), and inbound repair time ($\gamma_{1,1}$).}
  \label{fig2}
\end{figure}

The mean orbit waiting time of incoming calls is given in \figurename{~\ref{fig3}} as functions of parameters $\lambda$ and $\beta_1$. With respect to the benchmark, we observe the impact of server failure while serving incoming calls on the gradual increase in $E[W]$ in \figurename{\ref{fig3a}}. As the value of $\lambda$ rises from 1 to 2, $E[W]$ shows a percentage increase of almost $133\%$ for $(\beta_1,\beta_2) = (0.4,0.5)$ and nearly $111\%$ for $(\beta_1,\beta_2) = (0.8,0.5)$. Moreover, as the value of $\beta_1$ increases to $2$ in \figurename{\ref{fig3b}}, we observe that the average orbit delay grows exponentially with increase in $\mu_{1,1}$.

\figurename{~\ref{fig4}} shows the impact of $\lambda$ and $\beta_1$ on server availability. It is noteworthy that all three distributions exhibit the same results for different values of $\beta_1$ and $\beta_2$. Therefore, they have been demonstrated using a single plot. In absence of server failure, $P_a$ is obviously always equal to 1. However, as $\beta_1$ increases, the availability of the server to incoming calls reduces with rise in $\lambda$. For instance, at $\lambda\!=\!1.6$ in \figurename{~\ref{fig4a}}, $P_a$ falls by slightly less than $1\%$ as $\beta_1$ goes from 0.4 to 0.8. \figurename{\ref{fig4b}} further portrays the prominent effect of $\beta_1$ on server availability under varying first moment of service and repair times. Note that there is a steeper fall in $P_a$ as the service time of incoming calls increases. For instance, given $\beta_1=1$, $P_a$ drops by approximately $15\%$ as $\mu_{1,1}$ increases by a factor of 5. This measure further deteriorates by around $31.8\%$ as the value of $\beta_1$ rises to 2. The figure reveals that the probability of finding the server available is higher when $\mu_{1,1} < \gamma_{1,1}$ and is more likely to reduce with increase in the inbound service time $\mu_{1,1}$.
\begin{figure}[!t]
  \begin{center}
    \subfigure[Mean orbit waiting time versus $\lambda$.]{\label{fig3a}\includegraphics[scale=0.73]{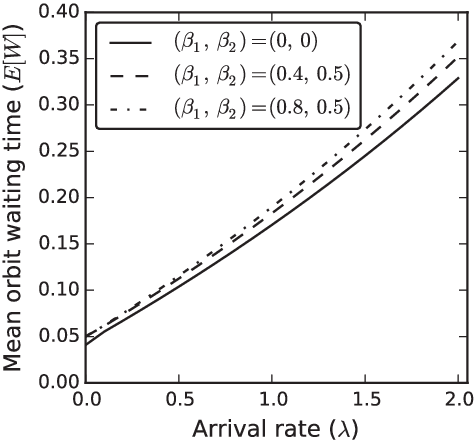}}~~~
    \subfigure[Mean orbit waiting time versus $\beta_1$.]{\label{fig3b}\includegraphics[scale=0.73]{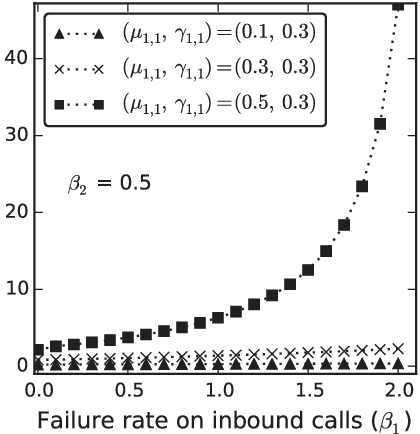}}
  \end{center}
  \caption{Expected waiting time in orbit ($E[W]$) versus inbound arrival rate ($\lambda$) and server failure rate while serving such calls ($\beta_1$).}
  \label{fig3}
\end{figure}
\begin{figure}[!t]
  \begin{center}
    \subfigure[Server availability versus $\lambda$.]{\label{fig4a}\includegraphics[scale=0.73]{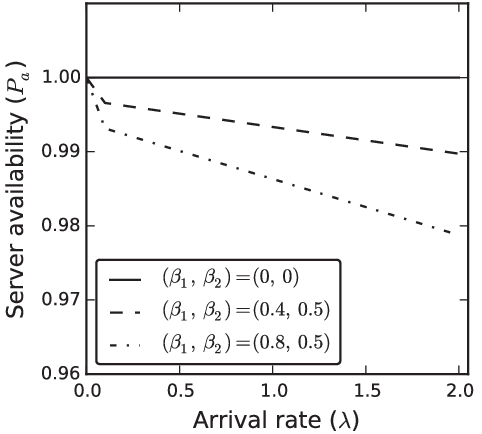}}~~~
    \subfigure[Server availability versus $\beta_1$.]{\label{fig4b}\includegraphics[scale=0.73]{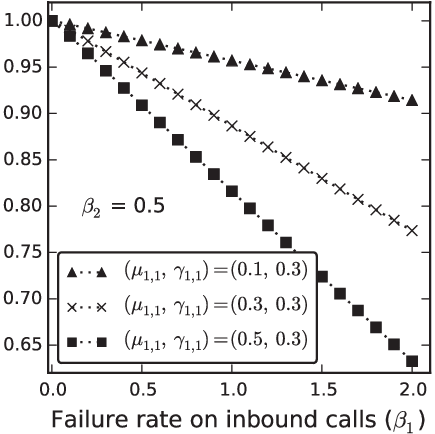}}
  \end{center}
  \caption{Server availability ($P_a$) as a function of the inbound arrival rate ($\lambda$) and the server failure rate while serving such calls ($\beta_1$).}
  \label{fig4}
\end{figure}

Similarly, \figurename{~\ref{fig5a}} depicts the server failure frequency as functions of parameters $\lambda$ and $\beta_1$. Apparent from (\ref{eq35}), we observe that $P_f$ monotonically increases with the number of incoming calls in our model. Additionally, at $\lambda\!=\!2$, as $\beta_1$ increases from 0.4 to 0.8, the value of $P_f$ rises drastically by over $74\%$. For various values of $(\mu_{1,1},\gamma_{1,1})$, \figurename{\ref{fig5b}} shows that $P_f$ constantly increases with $\beta_1$ and is higher when $\mu_{1,1}$ is more than $\gamma_{1,1}$. In comparison to \figurename{\ref{fig4}}, such behavior is not far from expectation as the measures $P_f$ and $P_a$ are inversely related.
\begin{figure}[!t]
  \begin{center}
    \subfigure[Server failure frequency versus $\lambda$.]{\label{fig5a}\includegraphics[scale=0.73]{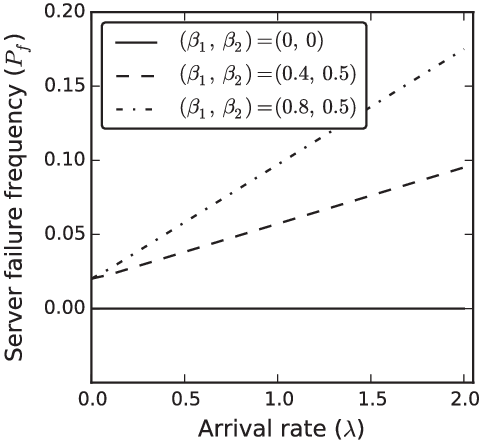}}~~~
    \subfigure[Server failure frequency versus $\beta_1$.]{\label{fig5b}\includegraphics[scale=0.73]{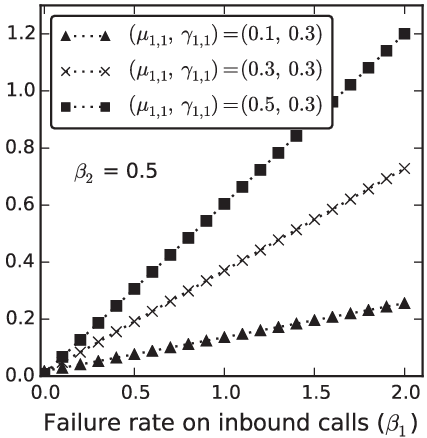}}
  \end{center}
  \caption{Server failure frequency ($P_f$) as a function of the inbound arrival rate ($\lambda$) and the server failure rate while serving such calls ($\beta_1$).}
  \label{fig5}
\end{figure}


\section{Conclusion}
\label{section6}
In this paper, we have conducted an exhaustive steady-state analysis of the $M/G/1$ retrial queue incorporated with two-way communication and the possibility of server failure. Having immediate applications in blended call centers, our study of the stationary characteristics provided explicit expressions for the joint distribution of the server state and the expected number of incoming customer calls in the orbit. Results from extensive numerical simulations were provided for various performance measures to validate and compare our findings with that of a baseline system with no server breakdown. A promising follow-up on this work would be an extended analysis involving an unreliable multi-server retrial queueing model with prioritized classes of incoming service requests. The consideration of customer impatience in conjunction with service prioritization is yet another interesting direction for further exploration. 




\end{document}